\documentclass[a4paper,twocolumn,11pt,unpublished]{quantumarticle}
\pdfoutput=1

\usepackage[utf8]{inputenc}
\usepackage[english]{babel}
\usepackage[T1]{fontenc}
\usepackage{graphicx}
\usepackage{amsmath,amssymb,amsfonts,mathtools,bm,amsthm}
\usepackage[numbers,sort&compress]{natbib}
\usepackage{booktabs}
\usepackage{microtype}
\usepackage{xcolor}
\usepackage[colorlinks=true,citecolor=blue,urlcolor=blue,linkcolor=blue]{hyperref}

\theoremstyle{plain}
\newtheorem{theorem}{Theorem}

\newcommand{\Tr}{\operatorname{Tr}}
\newcommand{\Inv}{\operatorname{Inv}}

\newcommand{\ED}{E_{\mathrm D}}
\newcommand{\EF}{E_{\mathrm F}}
\newcommand{\cH}{\mathcal H}
\newcommand{\cF}{\mathcal F}
\newcommand{\id}{\mathbb I}

\begin{document}

\title{Thermal Activation of Divergent Distillable Entanglement under Non-Abelian Strong Symmetry}
\author{Shuai Zeng}
\email{zengshuai@cqupt.edu.cn}
\affiliation{Chongqing University of Posts and Telecommunications, Chongqing 400065, China}
\maketitle

\begin{abstract}
Heating usually destroys quantum entanglement. We show that thermalization constrained to a non-Abelian strong-symmetry sector can instead generate a distillable resource that diverges with system size. In an exactly solvable local dimer chain, entanglement across an equal bipartition is exactly zero at $T=0$, whereas every fixed $T>0$ yields
$\ED=\frac12\log_2N+C(T)+o(1)$.
Measurements of the two half-chain representation labels convert thermally populated non-Abelian sectors into standard ebits. More generally, for global-singlet thermal states of finite-range, uniformly bounded, locally $SU(2)$-invariant chains, there is a nonzero high-temperature interval in which the protocol yield satisfies
$Y_N=\frac12\log_2N+O_\beta(1)$ and $\ED\geq Y_N$.
For the dimer chain, the full finite-size onset is governed by a universal Bessel-function crossover with $T_*(N)=\Delta/[\ln N+O(1)]$. Exact diagonalization of a frustrated $J_1$--$J_2$ chain shows the expected finite-size signatures. Thus thermal fluctuations can create entanglement across a macroscopic cut and convert it into an unbounded operational quantum resource.
\end{abstract}

\section{Introduction}\label{sec:intro}

Heating is normally an enemy of entanglement: increasing temperature mixes quantum states and erases the coherences needed for nonclassical correlations. Entanglement is an operational resource central to quantum information and a basic organizing quantity in many-body physics~\cite{Horodecki2009,Amico2008}. At finite temperature, locality strongly constrains bipartite entanglement and its spatial organization~\cite{Kuwahara2022}. Unconstrained Gibbs states of local spin systems become separable above a sufficiently high temperature~\cite{Bakshi2024}; in one dimension, the Schmidt number across a spatial cut remains finite at every positive temperature~\cite{Bakshi2026}, and sufficiently separated regions become separable~\cite{Scalet2026}. Finite-temperature entanglement can nevertheless persist, emerge or be enhanced in particular models~\cite{Arnesen2001,Vedral2004}. We ask a stronger question: can heating start from exactly zero entanglement across a macroscopic cut and create an LOCC-distillable resource that grows without bound with system size?

Non-Abelian strong symmetry makes this reversal possible by changing the equilibrium ensemble. Strong symmetry decomposes open-system dynamics into invariant sectors~\cite{Buca2012}; symmetry-preserving dynamics initialized in a fixed charge sector can therefore thermalize within that sector. Recent results show that fixed-charge ensembles can avoid high-temperature separability even for Abelian symmetries~\cite{Garratt2026,Negari2025}. Non-Abelian charges also reshape typical and symmetry-resolved entanglement~\cite{Majidy2023,Bianchi2024}. At infinite temperature, maximally mixed invariant sectors possess exactly distillable representation-space entanglement~\cite{Livine2005,Moharramipour2024}, while strongly symmetric stationary states exhibit related logarithmic entanglement structures~\cite{Li2025}. The remaining finite-temperature questions are whether local energetic weighting preserves the growing distillable resource and whether thermalization can generate it from a ground state that factorizes across the chosen cut.

The mechanism converts ordinary thermal fluctuations into representation-space entanglement. In an equal bipartition $A|B$ of a global $SU(2)$ singlet, the half-chain spin distribution has a characteristic width $j=O(\sqrt N)$. The singlet constraint locks the two halves into the same spin-$j$ representation, whose unique invariant vector is maximally entangled with Schmidt rank $2j+1$. Alice and Bob each measure the representation label of their half-chain and discard the associated multiplicity degrees of freedom, thereby extracting $\log_2(2j+1)$ ebits. These operations are local with respect to the bipartition, although the representation measurement can be collective within each half. The thermal representation width therefore produces an average yield of order $\frac12\log_2N$.

\subsection{Main results and scope}\label{sec:contributions}

\paragraph{Exact activation from zero entanglement.}
We construct a local dimer chain whose unique ground state factorizes across the central bipartition. Sector-preserving heating populates non-Abelian representation sectors, and an explicit LOCC protocol attains $\ED=\EF=Y_N$ exactly at every finite size. For every fixed $T>0$, $\ED=\frac12\log_2N+C(T)+o(1)$, whereas $\ED=0$ at $T=0$. The full dilute-to-thermodynamic crossover is a closed Bessel-function law controlled by $x=(N/4)e^{-\Delta/T}$, with onset scale $T_*(N)=\Delta/[\ln N+O(1)]$.

\paragraph{General local-chain theorem.}
The activation mechanism is not tied to the decoupled dimer limit. Its universal high-temperature content is summarized by the following theorem, stated here so that the assumptions and scaling are visible at the outset.

\begin{theorem}[Finite-temperature non-Abelian distillation]\label{thm:main}
Consider a one-dimensional chain with a fixed finite-period unit cell, a uniformly bounded finite-range interaction, exact local $SU(2)$ invariance and reflection symmetry about an equal bipartition. Let
$\rho_{\beta,0}=P_0e^{-\beta H_N}/\Tr(P_0e^{-\beta H_N})$
be the thermal state in the global-singlet sector, and let $Y_N$ be the yield of the local representation-space distillation protocol. There exists $\beta_0>0$ such that, for $0\leq\beta<\beta_0$ and all sufficiently large even $N$ for which the singlet sector is nonempty,
\begin{equation}
Y_N=\frac12\log_2N+O_\beta(1),
\qquad
\ED(\rho_{\beta,0})\geq Y_N,
\label{eq:intro-theorem}
\end{equation}
where the $O_\beta(1)$ term is independent of $N$.
\end{theorem}

\paragraph{Mechanism, proof and finite-size verification.}
The proof combines a uniform complex-twist cluster expansion, character estimates around the two central elements of $SU(2)$ and a quasi-local quantum-belief-propagation dressing of the interactions crossing the bipartition. These ingredients show that the $O(\sqrt N)$ representation width is a bulk fluctuation stable against cut-local perturbations. Exact diagonalization of a frustrated $J_1$--$J_2$ chain is consistent with an order-one residual $Y_N-\frac12\log_2N$ and an order-one distribution of $j/\sqrt N$ over the accessible sizes. The exact dimer result holds for every $T>0$; the general theorem establishes the universal coefficient in a finite high-temperature interval under the assumptions stated above.

\section{Results}\label{sec:results}

\subsection{Non-Abelian symmetry converts thermal fluctuations into distillable entanglement}\label{sec:mechanism}

Let a compact group act on the two halves as
\begin{equation}
 \cH_A=\bigoplus_{\lambda}V_{\lambda}\otimes M^A_{\lambda},
 \qquad
 \cH_B=\bigoplus_{\mu}V_{\mu}\otimes M^B_{\mu},
 \label{eq:general-decomp}
\end{equation}
where $V_{\lambda}$ is an irreducible representation and $M^{A,B}_{\lambda}$ is the associated multiplicity space. In the invariant subspace, Schur's lemma fixes the representation-space factor:
\begin{equation}
 \Inv(\cH_A\otimes\cH_B)
 =\bigoplus_{\lambda}
 \operatorname{span}\{|\Phi_{\lambda}\rangle\}
 \otimes M^A_{\lambda}\otimes M^B_{\bar\lambda}.
 \label{eq:inv-decomp}
\end{equation}
Here $|\Phi_{\lambda}\rangle$ is the unique normalized invariant vector in $V_{\lambda}\otimes V_{\bar\lambda}$ and is maximally entangled with Schmidt rank $d_{\lambda}=\dim V_{\lambda}$.

Alice and Bob locally measure their isotypic labels. The global invariant constraint makes the outcomes perfectly correlated as $(\lambda,\bar\lambda)$, and the measurement removes any coherence between distinct representation sectors. Conditioned on $\lambda$, discarding the multiplicity spaces leaves $|\Phi_{\lambda}\rangle$. Thus every state supported in the invariant subspace obeys
\begin{align}
 \ED(\rho_{AB})&\geq
 \sum_{\lambda}p_{\lambda}\log_2d_{\lambda},
 \label{eq:general-distill}\\
 p_{\lambda}&=\Tr[(P^A_{\lambda}\otimes P^B_{\bar\lambda})\rho_{AB}].
 \nonumber
\end{align}
The protocol consists of representation measurements local to each party, classical communication and local discarding, and is therefore an LOCC distillation protocol~\cite{Horodecki2009}. The symmetry determines the sector-resolved thermal state; the extraction itself uses ordinary local quantum operations.

For $SU(2)$, the irreducible labels are spins $j$ and $d_j=2j+1$. We define the representation-space yield
\begin{equation}
 Y_N=\sum_jp_j\log_2(2j+1),
 \qquad \ED(\rho_{AB})\geq Y_N.
 \label{eq:yield}
\end{equation}
Equation~\eqref{eq:yield} separates the problem into two components: a universal information-theoretic conversion from representation dimension to ebits, and a thermodynamic problem that determines the distribution $p_j$. Fig.~\ref{fig:mechanism} summarizes the conversion. Local thermal fluctuations create a half-chain spin width of order $\sqrt N$; the global singlet constraint locks the same representation across the bipartition; and the non-Abelian dimension $2j+1$ converts that width into a logarithmic distillation yield.

\begin{figure*}[t]
 \centering
 \includegraphics[width=\textwidth]{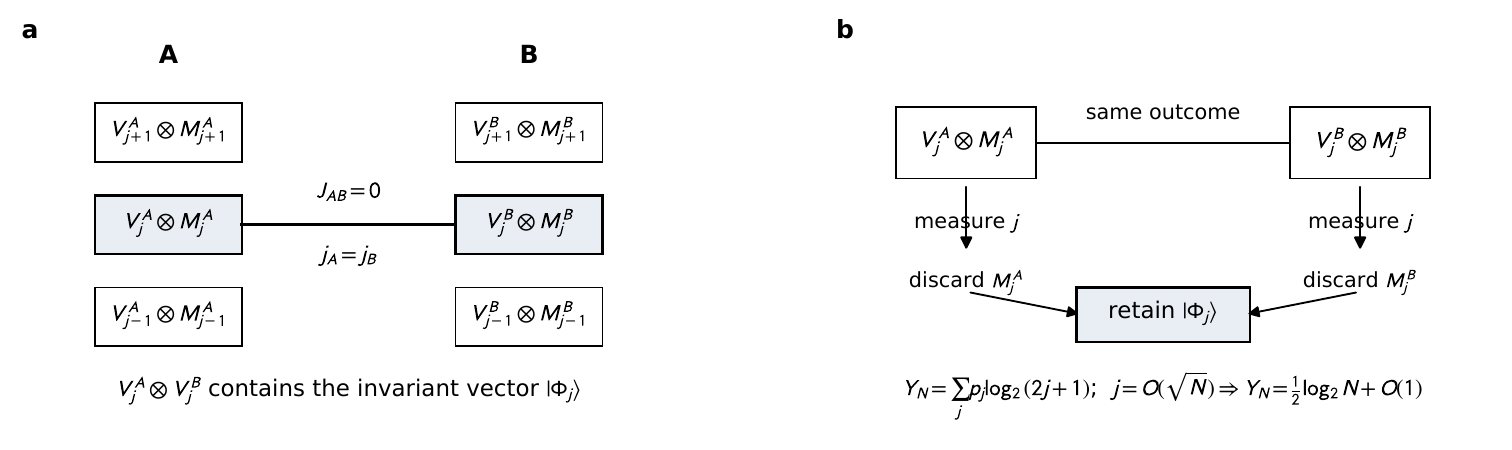}
 \caption{\textbf{Representation-space distillation.} \textbf{a}, The global singlet constraint correlates equal local $SU(2)$ representation labels across the bipartition. For each outcome $j$, the representation factors $V_j^A$ and $V_j^B$ contain the unique invariant vector $|\Phi_j\rangle$. \textbf{b}, Alice and Bob each measure $j$ on their half-chain and discard the multiplicity factors $M_j^A$ and $M_j^B$, retaining $|\Phi_j\rangle$ with $\log_2(2j+1)$ ebits. Thermal fluctuations with $j=O(\sqrt N)$ therefore give a yield of order $\frac12\log_2N$.}
 \label{fig:mechanism}
\end{figure*}

\subsection{Heating an unentangled ground state into a divergent resource}\label{sec:dimer-activation}

An exactly solvable dimer chain provides an explicit realization of the effect. We prepare the singlet ground state and then let a detailed-balance dynamics with strong $SU(2)$ symmetry equilibrate within the total-spin-zero sector fixed by that initial state. Take $N=4L$ spin-$1/2$ degrees of freedom and pair them into $2L$ disjoint dimers, with $L$ dimers entirely inside each half. The Hamiltonian is
\begin{equation}
 H_{\mathrm{dim}}=\Delta\sum_{r=1}^{2L}P_r^{(1)},
 \label{eq:dimer-H}
\end{equation}
where $P_r^{(1)}$ projects dimer $r$ onto its triplet. The unique ground state is a product of dimer singlets. Because no dimer crosses $A|B$, the ground state factorizes across the bipartition and
\begin{equation}
 \ED(T=0)=0.
 \label{eq:dimer-zero}
\end{equation}

At positive temperature, put $y=e^{-\beta\Delta}$. Each dimer contributes $V_0\oplus yV_1$, so the half-chain representation content is
\begin{equation}
 (V_0\oplus yV_1)^{\otimes L}
 =\bigoplus_ja_{L,j}(y)V_j,
 \label{eq:dimer-rep}
\end{equation}
with coefficients
\begin{equation}
 a_{L,j}(y)=\int_{SU(2)}dg\,\chi_j(g)^*[1+y\chi_1(g)]^L.
 \label{eq:dimer-character}
\end{equation}
After projection onto global spin zero, the state is locally flagged by $j$:
\begin{align}
 \rho_{\beta,0}&=\bigoplus_jp_j
 |\Phi_j\rangle\!\langle\Phi_j|\otimes
 \tau^A_{L,j}\otimes\tau^B_{L,j},
 \label{eq:dimer-flagged}\\
 p_j&=\frac{a_{L,j}(y)^2}{\sum_ka_{L,k}(y)^2}.
 \nonumber
\end{align}
The multiplicity states factor across the cut. The local $j$ measurement therefore attains $Y_N$, while a product-state decomposition of the multiplicity factors upper-bounds the entanglement of formation by the same quantity. Consequently,
\begin{equation}
 \ED=\EF=Y_N
 \label{eq:dimer-equality}
\end{equation}
for every finite $L$ and every temperature.

For every fixed $T>0$, a two-center group saddle gives
\begin{align}
 \ED={}&\frac12\log_2L
 +\frac12\log_2[2\kappa(\beta)]\nonumber\\
 &+\frac{1-\gamma/2}{\ln2}+o(1),
 \label{eq:dimer-fixedT}\\
 \kappa(\beta)={}&\frac{e^{-\beta\Delta}}{1+3e^{-\beta\Delta}}.
 \nonumber
\end{align}
Here $\gamma$ is Euler's constant. Thus an arbitrarily low but fixed positive temperature produces an unbounded resource in a sufficiently large chain, even though the zero-temperature state is unentangled across the same cut. The limits do not commute:
\begin{equation}
 \lim_{T\downarrow0}\lim_{N\to\infty}\ED=\infty,
 \qquad
 \lim_{N\to\infty}\lim_{T\downarrow0}\ED=0.
 \label{eq:noncommuting}
\end{equation}
Fig.~\ref{fig:dimer}a and c show the activation and the approach to the fixed-temperature asymptotics.

\subsection{Exact crossover from dilute activation to logarithmic entanglement}\label{sec:crossover}

The low-temperature onset is governed by the triplet activity
\begin{equation}
 x=Le^{-\Delta/T}.
 \label{eq:xdef}
\end{equation}
We take the double-scaling limit $L\to\infty$, $T\to0$ with $x$ fixed. Since
\begin{equation}
 \left[1+\frac{x}{L}\chi_1(g)\right]^L\longrightarrow e^{x\chi_1(g)},
 \label{eq:poisson-limit}
\end{equation}
the limiting character coefficient is
\begin{align}
 A_j(x)&=\int_{SU(2)}dg\,\chi_j(g)e^{x\chi_1(g)}
 \label{eq:Aj}\\
 &=e^x[I_j(2x)-I_{j+1}(2x)].
 \nonumber
\end{align}
Here $I_j$ is the modified Bessel function. Parseval's identity gives the normalized distribution
\begin{equation}
 p_j(x)=\frac{[I_j(2x)-I_{j+1}(2x)]^2}
 {I_0(4x)-I_1(4x)}.
 \label{eq:pjx}
\end{equation}
The complete crossover function is therefore
\begin{align}
 \cF(x)&=\sum_{j=0}^{\infty}p_j(x)\log_2(2j+1),
 \label{eq:crossover}\\
 \lim_{L\to\infty}Y_L(x/L)&=\cF(x).
 \nonumber
\end{align}
The convergence is locally uniform in $x$, and the scaling function is completely determined by the activity $x$.

The two asymptotic regimes expose the activation mechanism:
\begin{align}
 \cF(x)&=x^2\log_2 3+O(x^3), &&x\downarrow0,
 \label{eq:smallx}\\
 \cF(x)&=\frac12\log_2(2x)
 +\frac{1-\gamma/2}{\ln2}+o(1), &&x\to\infty.
 \label{eq:largex}
\end{align}
The quadratic small-$x$ law follows because the global singlet projection requires matching triplet representations in both halves; a single activated triplet cannot form a global singlet across the bipartition. The large-$x$ law joins continuously onto the fixed-temperature result in Eq.~\eqref{eq:dimer-fixedT}. The crossover to an order-one resource occurs at $x=O(1)$, giving
\begin{equation}
 \frac{\Delta}{T_*(L)}=\ln L+O(1),
 \qquad
 T_*(N)=\frac{\Delta}{\ln N+O(1)}.
 \label{eq:Tstar}
\end{equation}
Fig.~\ref{fig:dimer}b shows finite-size curves converging to the exact Bessel-function limit.

\begin{figure*}[t]
 \centering
 \includegraphics[width=\textwidth]{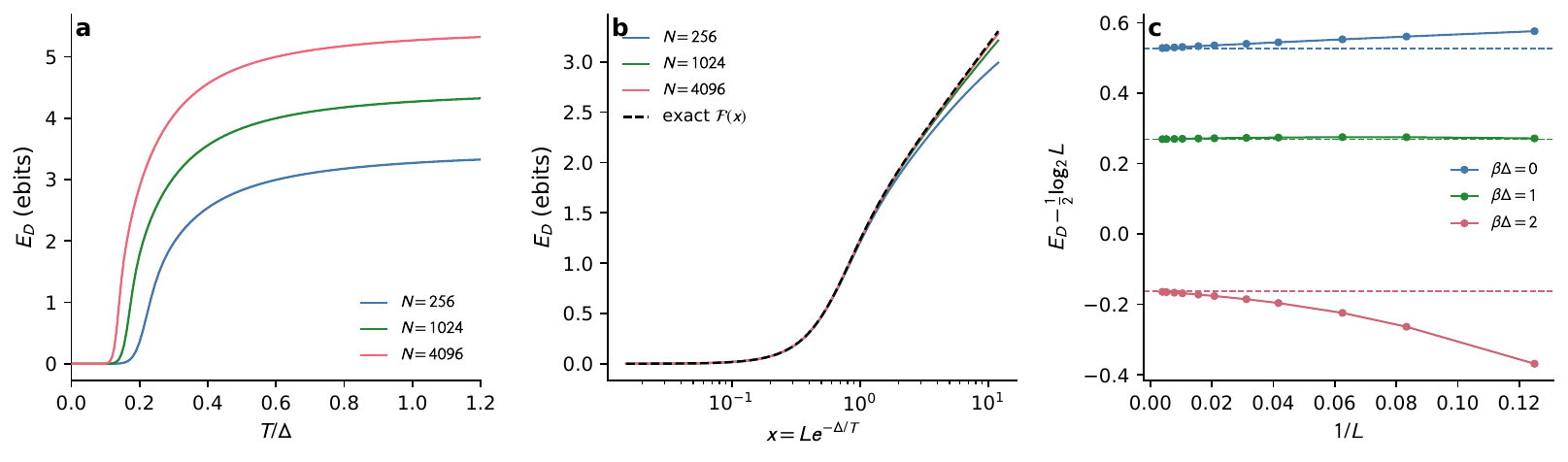}
 \caption{\textbf{Exact thermal activation and crossover in the dimer chain.} \textbf{a}, Distillable entanglement versus temperature for three system sizes. The curves begin at $E_{\mathrm D}=0$ because no dimer crosses the central bipartition, and rise with temperature toward a resource that grows with system size. \textbf{b}, Finite-size results converge to the exact crossover $\mathcal F(x)$ in Eq.~\eqref{eq:crossover}, shown as a black dashed line. \textbf{c}, The residual $E_{\mathrm D}-\frac12\log_2L$ approaches the analytic constants in Eq.~\eqref{eq:dimer-fixedT}; dashed horizontal lines show the asymptotic values.}
 \label{fig:dimer}
\end{figure*}

\subsection{Universal logarithmic entanglement in local spin chains}\label{sec:universal}

We now turn from the exactly solvable model to a broad class of local chains. Let $H_N$ be the restriction to $N$ sites of a fixed finite-period interaction on identical spins transforming in a nontrivial irreducible representation of $SU(2)$. After grouping one period into a unit cell, assume that every local term is $SU(2)$ invariant, that the interaction has finite range and a uniform local-norm bound, and that the finite chains are reflection symmetric about the central cut. The thermal state constrained to the global-singlet sector is
\begin{equation}
 \rho_{\beta,0}=\frac{P_0e^{-\beta H_N}}{Z_{\beta,0}},
 \qquad
 Z_{\beta,0}=\Tr(P_0e^{-\beta H_N}),
 \label{eq:singlet-gibbs}
\end{equation}
where $P_0$ projects onto total spin zero. Such a state is stationary under any symmetry-preserving detailed-balance dynamics that is initialized in and ergodic within the singlet sector.

Theorem~\ref{thm:main} shows that a finite-temperature local chain confined to the singlet sector contains a standard, explicitly extractable resource whose size diverges with the chain length. At $\beta=0$, it reduces to the logarithmic representation-space resource of the maximally mixed invariant sector~\cite{Moharramipour2024}. The finite-temperature advance is that arbitrary bounded finite-range interactions may reweight all multiplicity sectors, yet throughout a nonzero temperature interval they leave the leading coefficient unchanged. The universal coefficient $1/2$ follows from the three-dimensional $SU(2)$ fluctuation cloud and the growth $d_j=2j+1$ of its irreducible representations, independently of microscopic couplings.

The proof separates bulk fluctuations from the finite set of interactions crossing the bipartition. First delete those cut interactions, writing $H^{(0)}=H_A+H_B$. Reflection symmetry makes the two half-chain Hamiltonians unitarily equivalent. By Schur's lemma,
\begin{align}
 H_A&=\bigoplus_j\id_{V_j}\otimes K_{n,j},
 \label{eq:schur-H}\\
 a_{n,j}&=\Tr_{M_j}e^{-\beta K_{n,j}},
 \qquad n=N/2.
 \nonumber
\end{align}
Projection of the decoupled Gibbs operator onto the global singlet gives
\begin{equation}
 p_j^{(0)}=\frac{a_{n,j}^2}{\sum_ka_{n,k}^2}.
 \label{eq:pj0-results}
\end{equation}
A normalized $SU(2)$-twisted partition function,
\begin{equation}
 f_{n,\beta}(g)=\frac{\Tr(e^{-\beta H_A}U_n(g))}{\Tr e^{-\beta H_A}}
 =\frac{1}{Z_A}\sum_ja_{n,j}\chi_j(g),
 \label{eq:twist-results}
\end{equation}
encodes this distribution. A uniform high-temperature cluster expansion produces Gaussian saddles around the two central elements of $SU(2)$ and exponential suppression away from them. Character orthogonality and the group Laplacian then yield
\begin{align}
 \frac{\sum_ja_{n,j}^2}{Z_A^2}&\asymp n^{-3/2},
 \label{eq:norm-scaling}\\
 \frac{a_{n,j}}{Z_A}&\leq C(2j+1)n^{-3/2},
 \label{eq:coef-bound}\\
 \sum_jp_j^{(0)}j(j+1)&\leq Cn.
 \label{eq:casimir-bound}
\end{align}
The first two estimates give a cubic lower-tail bound
\begin{equation}
 \Pr_{p^{(0)}}(2j+1\leq x\sqrt n)\leq Cx^3,
 \qquad 0<x\leq1,
 \label{eq:low-tail-results}
\end{equation}
whereas Eq.~\eqref{eq:casimir-bound} controls the upper tail. Together they imply
\begin{equation}
 \sum_jp_j^{(0)}\log_2(2j+1)
 =\frac12\log_2n+O_{\beta}(1).
 \label{eq:decoupled-yield}
\end{equation}

The second step restores the cut interaction. Quantum belief propagation expresses the coupled Gibbs operator as a bounded, invertible dressing of the decoupled operator,
\begin{equation}
 e^{-\beta(H^{(0)}+V_{\partial})}
 =\eta e^{-\beta H^{(0)}}\eta^{\dagger}.
 \label{eq:qbp-results}
\end{equation}
The dressing can be replaced by an exactly $SU(2)$-invariant operator $\eta_R$ supported within distance $R$ of the cut, with exponentially small error. Such an operator changes the half-chain spin by at most $O(R)$. Choosing $R=q\log N$ makes the truncation error negligible while $R=o(\sqrt N)$. The cubic low-spin tail and the linear Casimir moment therefore survive the cut-local perturbation, preserving Eq.~\eqref{eq:decoupled-yield} and proving Theorem~\ref{thm:main}.

\subsection{Locality protects the thermal representation fluctuations}\label{sec:locality}

The proof also identifies locality as a physical component of the mechanism. The logarithmic resource is generated by the extensive bulk: at positive temperature, many local degrees of freedom contribute incoherently to the half-chain charge, producing a representation distribution of width $O(\sqrt N)$. The interactions that connect the two halves occupy only a fixed neighborhood of the cut. Their Gibbs dressing spreads quasi-locally, reaching only $O(\log N)$ sites at the accuracy needed for the thermodynamic estimate. A boundary operation with this support cannot reorganize a distribution whose intrinsic width is $O(\sqrt N)$.

Locality is essential to the mechanism. Nonlocal penalties such as $\bm S_A^2+\bm S_B^2$ act directly on the collective half-chain spin and can confine $j$ to $O(1)$, whereas local interactions retain the thermodynamic representation fluctuations that generate the logarithmic yield. The strong non-Abelian constraint correlates the two fluctuation clouds and converts their representation dimension into ebits.

\subsection{Finite-size verification in a frustrated spin chain}\label{sec:ed-results}

To illustrate these predictions away from the solvable limits, we examined the open $J_1$--$J_2$ spin-$1/2$ chain
\begin{equation}
 H=J_1\sum_i\bm S_i\cdot\bm S_{i+1}
 +J_2\sum_i\bm S_i\cdot\bm S_{i+2}.
 \label{eq:J1J2}
\end{equation}
We set $J_2/J_1=0.37$, a generic frustrated coupling. We diagonalized directly in the global-singlet sector for $N\leq16$ and thermally averaged the half-chain spin projectors. The theorem fixes the asymptotic coefficient; these system sizes test the corresponding finite-size signatures away from the solvable dimer limit.

Fig.~\ref{fig:ed}a shows that $Y_N$ grows along both half-chain parity subsequences. More directly, Fig.~\ref{fig:ed}b plots the residual $Y_N-\frac12\log_2N$, which remains within an order-one range but displays parity-dependent finite-size drift over the accessible sizes. Fig.~\ref{fig:ed}c shows cumulative distributions of $j/\sqrt N$ at $\beta J_1=0.5$; the distributions remain on an order-one horizontal scale as $N$ increases, consistent with the predicted $\sqrt N$ width. These data provide finite-size evidence for the representation-fluctuation mechanism away from the decoupled and exactly solvable limits.

\begin{figure*}[t]
 \centering
 \includegraphics[width=\textwidth]{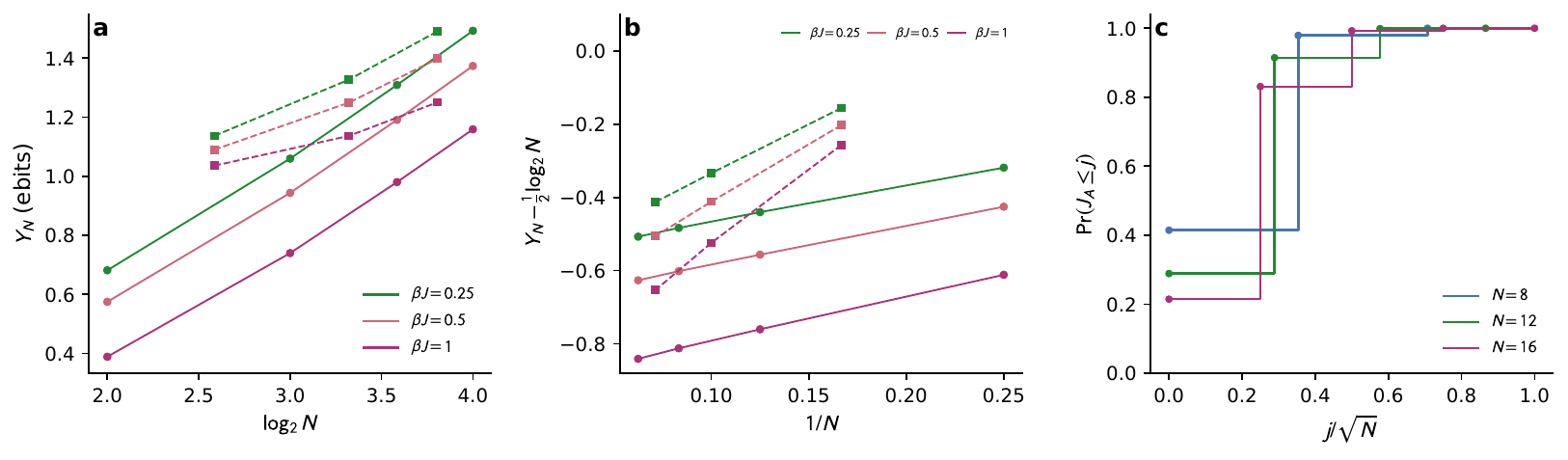}
 \caption{\textbf{Finite-size signatures in a frustrated spin chain.} Singlet-sector exact diagonalization of the open $J_1$--$J_2$ chain at $J_2/J_1=0.37$. \textbf{a}, Representation-space yield $Y_N$ versus $\log_2N$. Solid circles connect $N=4m$, and dashed squares connect $N=4m+2$. \textbf{b}, The residual $Y_N-\frac12\log_2N$ remains of order one over the accessible sizes and displays parity-dependent finite-size drift around the theorem-fixed leading behavior. \textbf{c}, At $\beta J_1=0.5$, cumulative distributions of $j/\sqrt N$ exhibit the predicted order-one rescaled width. Connecting lines are guides to the eye.}
 \label{fig:ed}
\end{figure*}

\section{Discussion}\label{sec:discussion}

Sector-preserving heating can therefore create entanglement across a macroscopic bipartition from a ground state that is exactly unentangled across that cut. The generated resource is operational: local representation measurements distill a number of ebits that grows without bound with system size. In the solvable dimer chain the activation occurs for every fixed positive temperature and is described by an exact crossover function; in the general local class Theorem~\ref{thm:main} fixes the leading high-temperature growth to $\frac12\log_2N$.

The result builds on several complementary strands of symmetry-constrained entanglement: non-Abelian charges can enhance typical entanglement~\cite{Majidy2023}, symmetry-resolved formulations isolate representation sectors~\cite{Bianchi2024}, maximally mixed invariant sectors contain exactly distillable representation-space entanglement~\cite{Livine2005,Moharramipour2024}, and strong symmetries generate highly entangled stationary states~\cite{Li2025}. We extend this picture in two directions. Reweighting by an arbitrary uniformly bounded finite-range Hamiltonian preserves the logarithmic coefficient throughout a nonzero temperature interval, while an explicit local model turns temperature into an activation parameter: the resource begins at exactly zero, follows a universal dilute-to-thermodynamic crossover and diverges for every fixed $T>0$.

The mechanism separates into three ingredients. Strong symmetry fixes the global charge sector, locality produces $\sqrt N$-scale subsystem charge fluctuations, and non-Abelian representation theory converts those fluctuations into a growing Schmidt rank. The last step distinguishes this mechanism from the high-temperature entanglement protected by Abelian strong symmetry~\cite{Garratt2026,Negari2025}: one-dimensional Abelian irreducible representations do not supply the representation-dimension contribution in Eq.~\eqref{eq:general-distill}.

The theorem assumes exact local $SU(2)$ invariance, preparation and thermalization within the global-singlet sector, a uniformly bounded finite-range interaction and reflection symmetry about the equal bipartition; its general conclusion is established in a high-temperature interval, whereas the exact dimer result holds at all positive temperatures. These assumptions identify concrete extensions: higher-rank compact groups, controlled symmetry breaking and imperfect sector preparation, restricted local control, and protocols that resolve non-Abelian charges without a full Schur transform. The present construction establishes that thermal fluctuations need not merely preserve constrained entanglement---they can generate an unbounded distillable resource from a cut-unentangled state.

\section{Methods}\label{sec:methods}

\subsection{Strong-symmetry thermal ensemble}\label{meth:ensemble}

We set Boltzmann's constant to one, so $\beta=1/T$. We consider finite chains with Hilbert space $\cH_N=(V_s)^{\otimes N}$, where $V_s$ is a fixed nontrivial spin-$s$ representation. The Hamiltonian is a finite-range interaction
\begin{equation}
 H_N=\sum_{X\subset[1,N]}\Phi_X,
 \label{eq:local-H}
\end{equation}
with a fixed finite-period unit cell and a uniform bound on the local interaction norm. Writing $U_N(g)=U_s(g)^{\otimes N}$, exact local $SU(2)$ invariance means $[\Phi_X,U_N(g)]=0$ for every interaction term. The central bipartition contains $n=N/2$ sites in each half, and reflection symmetry about the cut identifies the thermal multiplicity weights of the two halves. The global-singlet canonical state is Eq.~\eqref{eq:singlet-gibbs}.

A dynamical realization is a quantum Markov process with strong $SU(2)$ symmetry in the sense that its Hamiltonian and jump operators commute with the global group action~\cite{Buca2012}. If its restriction to the singlet sector is primitive and satisfies detailed balance with $\rho_{\beta,0}$, then an initial state in that sector converges to $\rho_{\beta,0}$. Our results concern the sector-resolved thermal state itself and therefore apply to any thermalizing generator with these symmetry, ergodicity and detailed-balance properties. The theorem is restricted to system sizes for which the singlet sector is nonempty.

\subsection{Representation-space distillation protocol}\label{meth:distill}

Let $P^A_j$ and $P^B_k$ be the projectors onto the local isotypic components in Eq.~\eqref{eq:general-decomp}. For a state $\rho=P_0\rho P_0$, Alice and Bob apply the local projective measurements $\{P^A_j\}$ and $\{P^B_k\}$. The invariant-subspace decomposition in Eq.~\eqref{eq:inv-decomp} implies that only equal outcomes $j=k$ occur. This remains true when $\rho$ contains coherences between different $j$ sectors, because the measurement removes them.

Conditioned on $j$, the normalized post-measurement state has the form
\begin{equation}
 |\Phi_j\rangle\!\langle\Phi_j|\otimes\sigma_j,
 \label{eq:conditional-state}
\end{equation}
where $\sigma_j$ is an arbitrary state on $M_j^A\otimes M_j^B$. Alice and Bob trace out their multiplicity systems and retain $|\Phi_j\rangle$, which contains $\log_2(2j+1)$ ebits. Repeating the protocol over many copies and classically sorting the outcomes yields the asymptotic average rate in Eq.~\eqref{eq:yield}. Appendix~\ref{app:distillation} gives the compact-group statement and proof.

\subsection{Representation distribution for decoupled halves}\label{meth:decoupled}

Remove all terms whose support intersects both halves and write $H^{(0)}=H_A+H_B$. Since $H_A$ commutes with the half-chain $SU(2)$ action,
\begin{equation}
 e^{-\beta H_A}=\bigoplus_j\id_{V_j}\otimes e^{-\beta K_{n,j}}.
 \label{eq:thermal-schur}
\end{equation}
We define the nonnegative multiplicity weights $a_{n,j}=\Tr_{M_j}e^{-\beta K_{n,j}}$. Reflection symmetry makes the corresponding weights on $B$ identical. Projecting $e^{-\beta H_A}\otimes e^{-\beta H_B}$ onto the global singlet retains one invariant vector in $V_j^A\otimes V_j^B$ and gives Eq.~\eqref{eq:pj0-results}. This identity is exact at finite size. Appendix~\ref{app:pressure} gives the projection identity and the associated character expansion.

\subsection{Twisted partition function and analytic pressure}\label{meth:twist}

Let $U_n(g)=U_s(g)^{\otimes n}$ and define
\begin{align}
 F_{n,\beta}(g)&=\Tr[e^{-\beta H_A}U_n(g)],
 \label{eq:Fdef}\\
 f_{n,\beta}(g)&=F_{n,\beta}(g)/Z_A.
 \nonumber
\end{align}
Its character expansion is $F_{n,\beta}(g)=\sum_ja_{n,j}\chi_j(g)$. In exponential coordinates $g=e^{i\bm X\cdot\bm S}$ around the identity, local $SU(2)$ invariance allows the twist to be written inside the exponential:
\begin{align}
 F_{n,\beta}(e^{i\bm X\cdot\bm S})
 &=\Tr\exp[-\beta H_A+i\bm X\cdot\bm M_n],
 \label{eq:complex-twist}\\
 \bm M_n&=\sum_{x=1}^n\bm S_x.
 \nonumber
\end{align}
High-temperature cluster expansions for finite-range quantum interactions~\cite{Netocny2004,Nguyen2024} extend to this bounded complex one-site source. The uniform version required here is established in Appendix~\ref{app:pressure}. Consequently,
\begin{equation}
 \log f_{n,\beta}(e^{i\bm X\cdot\bm S})
 =n\phi_{\beta}(\bm X)+b_{n,\beta}(\bm X),
 \label{eq:pressure}
\end{equation}
where $\phi_{\beta}$ is an analytic bulk pressure and $b_{n,\beta}$ is an analytic boundary term whose derivatives through third order are bounded uniformly in $n$. Rotational invariance forces the Hessian at the origin to be isotropic,
\begin{equation}
 -\partial_a\partial_b\phi_{\beta}(0)=\sigma_{\beta}^2\delta_{ab},
 \qquad \sigma_{\beta}^2>0,
 \label{eq:susceptibility}
\end{equation}
through a nonzero high-temperature interval. The same expansion applies around the second central element. Away from fixed neighborhoods of the two centers, a support-cover estimate gives $|f_{n,\beta}(g)|\leq Ce^{-cn}$. Appendices~\ref{app:pressure} and~\ref{app:twisted-bounds} prove the uniform complex-twist pressure statement and derive the global Gaussian and away-from-center bounds.

\subsection{Character bounds and logarithmic yield}\label{meth:characters}

Character orthogonality gives
\begin{equation}
 \frac{\sum_ja_{n,j}^2}{Z_A^2}
 =\int_{SU(2)}dg\,|f_{n,\beta}(g)|^2.
 \label{eq:parseval}
\end{equation}
The two Gaussian saddles in three group dimensions produce the $n^{-3/2}$ normalization in Eq.~\eqref{eq:norm-scaling}. Fourier inversion of the local Gaussian gives the coefficient bound in Eq.~\eqref{eq:coef-bound}. Applying the group Laplacian and using $-\Delta_G\chi_j=j(j+1)\chi_j$ yields
\begin{equation}
 \sum_j j(j+1)\frac{a_{n,j}^2}{Z_A^2}
 =\int dg\,\|\nabla f_{n,\beta}(g)\|^2
 =O(n^{-1/2}),
 \label{eq:laplacian}
\end{equation}
which becomes Eq.~\eqref{eq:casimir-bound} after normalization.

Summing the squared coefficient bound for $2j+1\leq x\sqrt n$ gives Eq.~\eqref{eq:low-tail-results}. Tail integration shows that the negative part of
\begin{equation}
 \log\frac{2j+1}{\sqrt n}
 \label{eq:scaled-log}
\end{equation}
is uniformly integrable. The Casimir moment controls the positive tail. The expectation of Eq.~\eqref{eq:scaled-log} is therefore $O(1)$, proving Eq.~\eqref{eq:decoupled-yield}. Appendix~\ref{app:characters} gives the estimates explicitly.

\subsection{Stability under the cut interaction}\label{meth:qbp}

Let $H(s)=H^{(0)}+sV_{\partial}$ for $s\in[0,1]$, and denote the normalized singlet-sector Gibbs states at $s=0$ and $s=1$ by $\rho_0$ and $\rho$, respectively. Quantum belief propagation gives an ordered exponential $\eta$ satisfying Eq.~\eqref{eq:qbp-results}, with $\|\eta\|$ and $\|\eta^{-1}\|$ bounded independently of $N$~\cite{Hastings2007,Capel2025}. The generator has exponentially decaying tails away from the support of $V_{\partial}$. Truncating it at distance $R$ gives $\eta_R$ with
\begin{equation}
 \|\eta-\eta_R\|\leq Ce^{-\mu R}.
 \label{eq:qbp-trunc}
\end{equation}
Averaging $\eta_R$ over the global $SU(2)$ action preserves its support and error while making it exactly symmetric.

An operator supported on $O(R)$ sites adjacent to the cut transforms as a direct sum of spins no larger than $cR$. The Wigner--Eckart selection rule therefore makes it banded in the half-chain spin:
\begin{equation}
 P^A_j\eta_RP^A_k=0
 \quad\text{when}\quad |j-k|>cR.
 \label{eq:banded}
\end{equation}
The bounded invertibility of $\eta_R$ compares the normalization of the coupled and decoupled singlet states. Equation~\eqref{eq:banded} transfers the decoupled low-spin probability to a threshold shifted by $cR$. A corresponding banded-operator estimate transfers the Casimir moment,
\begin{align}
 \Tr[(2J_A+1)^2\rho]
 &\leq C\big\{\Tr[(2J_A+1)^2\rho_0]+R^2
 \nonumber\\
 &\hspace{4.7em}+N^2e^{-\mu R}\big\}.
 \label{eq:moment-transfer}
\end{align}
Choosing $R=q\log N$ with sufficiently large $q$ makes the final term negligible and keeps $R/\sqrt N\to0$. The logarithmic-yield proof then applies to the coupled state. Appendix~\ref{app:stability} supplies the normalization, tail and moment inequalities.

\subsection{Exact dimer decomposition and fixed-temperature asymptotics}\label{meth:dimer}

For one dimer, the thermal representation polynomial is $V_0\oplus yV_1$. Tensoring $L$ copies gives Eq.~\eqref{eq:dimer-rep}; the coefficients also obey the recursion
\begin{align}
 a^{(\ell+1)}_0&=a^{(\ell)}_0+y a^{(\ell)}_1,\nonumber\\
 a^{(\ell+1)}_j&=a^{(\ell)}_j
 +y\big(a^{(\ell)}_{j-1}+a^{(\ell)}_j+a^{(\ell)}_{j+1}\big),
 \qquad j\geq1,
 \label{eq:dimer-recursion}
\end{align}
with $a^{(0)}_0=1$. This stable recursion generates the finite-size data in Fig.~\ref{fig:dimer}.

For the analytic limit, normalize the one-dimer character as
\begin{equation}
 q_y(g)=\frac{1+y\chi_1(g)}{1+3y}.
 \label{eq:qy}
\end{equation}
Its modulus reaches one only at the two central elements. In exponential coordinates,
\begin{equation}
 q_y(e^X)=1-\kappa\|X\|^2+O(\|X\|^4),
 \qquad \kappa=\frac{y}{1+3y}.
 \label{eq:q-saddle}
\end{equation}
The character coefficient has the uniform saddle form, for $j=O(\sqrt L)$,
\begin{align}
 \frac{a_{L,j}}{(1+3y)^L}
 & =C_yL^{-3/2}(2j+1)
 \label{eq:dimer-saddle}\\
 &\quad\times\exp\left[-\frac{j(j+1)}{4\kappa L}\right][1+o(1)].
 \nonumber
\end{align}
After squaring and normalizing, the variable
\begin{equation}
 U_L=\frac{2j+1}{\sqrt{8\kappa L}}
 \label{eq:UL}
\end{equation}
converges to the density $h(u)=4u^2e^{-u^2}/\sqrt\pi$. The coefficient and Casimir bounds make $\log U_L$ uniformly integrable. Since $\mathbb E_h\log U=\frac12\psi(3/2)=1-\gamma/2-\ln2$, Eq.~\eqref{eq:dimer-fixedT} follows. Appendix~\ref{app:dimer} gives the finite-size equality and the tail estimates.

\subsection{Dilute-activity crossover}\label{meth:crossover}

Let $\nu_{m,j}$ be the multiplicity of $V_j$ in $V_1^{\otimes m}$. Expanding the finite-$L$ coefficient at $y=x/L$ gives
\begin{equation}
 a_{L,j}(x/L)=\sum_{m=0}^{L}\binom{L}{m}(x/L)^m\nu_{m,j}.
 \label{eq:binomial-cross}
\end{equation}
For fixed $m$, the binomial coefficient converges to $x^m/m!$. Since $\nu_{m,j}=0$ for $j>m$ and $\nu_{m,j}\leq3^m$, the factorial tail provides a summable bound even after multiplication by $1+\log(2j+1)$. Dominated convergence proves Eq.~\eqref{eq:crossover} locally uniformly in $x$.

For conjugacy angle $\theta$,
\begin{align}
 dg&=\pi^{-1}\sin^2(\theta/2)d\theta,
 \label{eq:class-coordinates}\\
 \chi_j(\theta)&=\frac{\sin[(2j+1)\theta/2]}{\sin(\theta/2)},
 \nonumber\\
 \chi_1(\theta)&=1+2\cos\theta.
 \nonumber
\end{align}
Substitution into Eq.~\eqref{eq:Aj} gives the Bessel difference. Parseval gives
\begin{equation}
 \sum_jA_j(x)^2=e^{2x}[I_0(4x)-I_1(4x)],
 \label{eq:bessel-normalization}
\end{equation}
which yields Eq.~\eqref{eq:pjx}. Taylor expansion gives Eq.~\eqref{eq:smallx}; the same two-center saddle used above, with $x$ replacing $\kappa L$, gives Eq.~\eqref{eq:largex}. Appendix~\ref{app:crossover} contains the full convergence and asymptotic proof.

\subsection{Exact diagonalization}\label{meth:ed}

We used the open Hamiltonian in Eq.~\eqref{eq:J1J2}. For $N\leq14$, we constructed the $S^z=0$ computational basis, diagonalized total $\bm S^2$ and retained its zero eigenspace. For $N=16$, we used the linearly independent noncrossing valence-bond basis and solved the Hamiltonian and $\bm S_A^2$ as generalized eigenvalue problems with the valence-bond Gram matrix. In each energy eigenstate in the singlet sector, we evaluated the spectral projectors of the half-chain Casimir $\bm S_A^2$ and then thermally averaged them to obtain $p_j$. The yield was computed from Eq.~\eqref{eq:yield}.

The distributed data were cross-validated in two independent representations for $N\leq14$. At $\beta=0$, the complete $p_j$ distribution through $N=16$ was additionally checked against the exact spin-$1/2$ multiplicity formula. The two half-chain parities were retained separately because they have distinct finite-size corrections. Appendix~\ref{app:ed} gives the basis construction and numerical validation details.

\section*{Data availability}
All processed numerical data supporting the figures are included as ancillary files with the arXiv record. The analytical curves are generated directly from the equations and procedures reported in the manuscript and Supplemental Material.

\section*{Code availability}
Python code that reproduces all three main-text figures in high-resolution PNG and vector-PDF formats from the analytical expressions and distributed numerical tables is included as ancillary files with the arXiv record. The exact-diagonalization code used to generate the processed numerical tables is available from the author upon reasonable request.

\section*{Acknowledgements}
This work was supported by the CPS-Yangtze Delta Region Industrial Innovation Center of Quantum and Information Technology-MindSpore Quantum Open Fund. The funder had no role in the study design, data generation or analysis, interpretation of the results, or preparation of the manuscript.

\section*{Author contributions}
S.Z. conceived the project, developed the analytical theory, performed the numerical calculations, wrote the code and prepared the manuscript.

\section*{Use of large language models}
Large language models were used for language editing and formatting assistance. The author independently verified all mathematical statements, proofs, numerical results and references and is fully responsible for the manuscript.

\section*{Competing interests}
The author declares no financial or non-financial competing interests.

\bibliographystyle{quantum}
\bibliography{references}

\clearpage
\onecolumn

\begin{center}
{\Large\bfseries Supplemental Material for\\[0.35em]
``Thermal Activation of Divergent Distillable Entanglement under Non-Abelian Strong Symmetry''\par}
\vspace{0.8em}
{\large Shuai Zeng\par}
\vspace{0.25em}
{\normalsize Chongqing University of Posts and Telecommunications, Chongqing 400065, China\par}
\end{center}
\vspace{1em}

\appendix
\numberwithin{equation}{section}
\setcounter{equation}{0}

This Supplemental Material provides the extended proofs and numerical construction supporting the main text. The main text states the physical models, operational protocol, principal estimates and computational methods; the derivations below follow the order of the corresponding Methods subsections. Logarithms are base two when they quantify entanglement and natural otherwise.

\section{Compact-group representation-space distillation}\label{app:distillation}

Let a compact group $G$ act on the two parties as
\begin{equation}
 \cH_A=\bigoplus_\lambda V_\lambda\otimes M^A_\lambda,
 \qquad
 \cH_B=\bigoplus_\mu V_\mu\otimes M^B_\mu.
 \label{S:eq:repdec}
\end{equation}
Schur's lemma gives
\begin{equation}
 \Inv(V_\lambda\otimes V_\mu)=0\quad(\mu\neq\bar\lambda),
 \qquad
 \dim\Inv(V_\lambda\otimes V_{\bar\lambda})=1.
\end{equation}
Denote the normalized invariant vector by $|\Phi_\lambda\rangle$.  In dual bases it is maximally entangled with Schmidt rank $d_\lambda=\dim V_\lambda$.  Hence
\begin{equation}
 \Inv(\cH_A\otimes\cH_B)=
 \bigoplus_\lambda \operatorname{span}\{|\Phi_\lambda\rangle\}
 \otimes M^A_\lambda\otimes M^B_{\bar\lambda}.
 \label{S:eq:invdec}
\end{equation}

\textbf{Lemma A.1.}  If $\rho_{AB}=P_{\rm inv}\rho_{AB}P_{\rm inv}$, then
\begin{equation}
 \ED(\rho_{AB})\geq\sum_\lambda p_\lambda\log_2d_\lambda,
 \qquad
 p_\lambda=\Tr[(P^A_\lambda\otimes P^B_{\bar\lambda})\rho_{AB}].
 \label{S:eq:EDbound}
\end{equation}

\emph{Proof.}
Alice and Bob perform their local isotypic measurements $\{P^A_\lambda\}$ and $\{P^B_\mu\}$.  Equation~\eqref{S:eq:invdec} makes the outcomes perfectly correlated as $(\lambda,\bar\lambda)$.  This measurement also removes possible coherences between distinct $\lambda$.  Conditioned on $\lambda$, the state is necessarily
\begin{equation}
 |\Phi_\lambda\rangle\!\langle\Phi_\lambda|\otimes\sigma_\lambda,
\end{equation}
with an arbitrary state $\sigma_\lambda$ on the multiplicity spaces.  Locally discarding those spaces leaves $|\Phi_\lambda\rangle$.  Applying this protocol independently to many copies and sorting the classical outcomes distills the average in Eq.~\eqref{S:eq:EDbound}.  \hfill$\square$

\section{Decoupled halves and uniform analytic pressure}\label{app:pressure}

Let the chain have $N=2n$ sites and first remove all interactions crossing the central cut, so $H^{(0)}=H_A+H_B$.  Reflection symmetry makes $H_A$ and $H_B$ unitarily equivalent.  Since $H_A$ commutes with the half-chain $SU(2)$ action,
\begin{equation}
 H_A=\bigoplus_j I_{V_j}\otimes K_{n,j},
 \qquad
 a_{n,j}=\Tr_{M_j}e^{-\beta K_{n,j}}\geq0.
 \label{S:eq:aj}
\end{equation}
Projecting $e^{-\beta H^{(0)}}$ onto the global singlet and tracing within each locally distinguishable block gives
\begin{equation}
 p_j^{(0)}=\frac{a_{n,j}^2}{D_n},\qquad D_n=\sum_{k}a_{n,k}^2.
 \label{S:eq:pj}
\end{equation}

Let $U_n(g)=U_s(g)^{\otimes n}$ be the half-chain representation and introduce the normalized twisted partition function
\begin{equation}
 F_{n,\beta}(g)=\Tr(e^{-\beta H_A}U_n(g)),
 \qquad f_{n,\beta}(g)=\frac{F_{n,\beta}(g)}{Z_A},
 \qquad Z_A=\Tr e^{-\beta H_A}.
 \label{S:eq:twist}
\end{equation}
It is a class function.  The block form in Eq.~\eqref{S:eq:aj} gives its character expansion
\begin{equation}
 F_{n,\beta}(g)=\sum_j a_{n,j}\chi_j(g),
 \qquad
 a_{n,j}=\int_{SU(2)}dg\,\chi_j(g)^*F_{n,\beta}(g).
 \label{S:eq:characterexpansion}
\end{equation}

The saddle analysis uses a high-temperature cluster expansion that remains uniform after adding the complex one-site source generating the $SU(2)$ twist. The required statement is recorded below.

\textbf{Lemma B.1 (uniform complex-twist pressure).}
Let $H_n=\sum_{X\subset[1,n]}\Phi_X$ be the restriction of a fixed one-dimensional,
finite-range, bounded interaction with a fixed finite-period unit cell.  Assume that every
$\Phi_X$ is $SU(2)$ invariant and that each site transforms in a fixed spin-$s$
irreducible representation.  Put $\bm M_n=\sum_{x=1}^n\bm S_x$.  There exist
$\beta_{\rm ce}>0$ and $\rho>0$ such that, for $0\leq\beta<\beta_{\rm ce}$ and
$\bm z\in\mathbb C^3$ with $\|\bm z\|<\rho$,
\begin{equation}
 L_{n,\beta}(\bm z)
 :=\log\frac{\Tr\exp[-\beta H_n+i\bm z\cdot\bm M_n]}
 {\Tr e^{-\beta H_n}}
 =n\phi_\beta(\bm z)+b_{n,\beta}(\bm z),
 \label{S:eq:complex-pressure}
\end{equation}
where the branch is fixed by $L_{n,\beta}(0)=0$.  The functions
$\phi_\beta$ and $b_{n,\beta}$ are analytic on the same polydisc and, for every
multi-index $\alpha$ with $|\alpha|\leq3$,
\begin{equation}
 \sup_{0\leq\beta<\beta_{\rm ce}}
 \sup_{n\geq1}\sup_{\|\bm z\|<\rho}
 \left|\partial_{\bm z}^{\alpha}b_{n,\beta}(\bm z)\right|<\infty,
 \qquad
 \sup_{0\leq\beta<\beta_{\rm ce}}
 \sup_{\|\bm z\|<\rho}
 \left|\partial_{\bm z}^{\alpha}\phi_\beta(\bm z)\right|<\infty.
 \label{S:eq:complex-pressure-derivatives}
\end{equation}
Moreover, $\phi_\beta(0)=b_{n,\beta}(0)=0$, both functions are invariant under the
adjoint $SO(3)$ action, and
\begin{equation}
 \partial_a\partial_b\phi_\beta(0)=-\kappa_\beta\delta_{ab},
 \qquad
 \kappa_\beta=\lim_{n\to\infty}\frac{1}{3n}
 \langle\bm M_n^2\rangle_{\beta,c}.
 \label{S:eq:complex-pressure-hessian}
\end{equation}
After reducing $\beta_{\rm ce}$ if necessary, $\kappa_\beta\geq s(s+1)/6>0$.

\emph{Proof.}
Because $[H_n,M_n^a]=0$ for each $a$,
\begin{equation}
 \Tr(e^{-\beta H_n}e^{i\bm z\cdot\bm M_n})
 =\Tr\exp[-\beta H_n+i\bm z\cdot\bm M_n].
\end{equation}
Treat the exponent as a dimensionless complex interaction
\begin{equation}
 K_{\beta,\bm z}(X)=
 \begin{cases}
  \beta\Phi_{\{x\}}-i\bm z\cdot\bm S_x,&X=\{x\},\\
  \beta\Phi_X,&|X|\geq2.
 \end{cases}
 \label{S:eq:dimensionless-interaction}
\end{equation}
The polymer expansion for finite-spin quantum lattice systems with multi-body
interactions is absolutely convergent whenever the corresponding weighted interaction
norm is sufficiently small.  For the family in Eq.~\eqref{S:eq:dimensionless-interaction},
that norm is bounded by $C_1\beta J_{\rm loc}+C_2s\|\bm z\|$.  Hence one may choose
$\beta_{\rm ce}$ and a larger polydisc of radius $\rho_+>\rho$ so that the expansion
converges uniformly for $\beta<\beta_{\rm ce}$ and $\|\bm z\|<\rho_+$.
The expansion is algebraic in the interaction matrices and therefore applies to their
complexification; absolute convergence supplies a nonzero finite-volume partition
function and a consistent analytic logarithm.

Subtracting the expansion at $\bm z=0$ gives
\begin{equation}
 L_{n,\beta}(\bm z)=\sum_{\Gamma\cap[1,n]\neq\varnothing}
 \widehat w_{\beta,\bm z}(\Gamma),
 \qquad
 \widehat w_{\beta,0}(\Gamma)=0,
 \label{S:eq:connected-twist-expansion}
\end{equation}
with an exponentially summable connected-cluster bound.  Uniform convergence on the
larger polydisc and Cauchy's estimate imply, for $|\alpha|\leq3$,
\begin{equation}
 \sup_x\sum_{\Gamma\ni x}e^{\mu\operatorname{diam}\Gamma}
 \sup_{\|\bm z\|<\rho}
 |\partial_{\bm z}^{\alpha}\widehat w_{\beta,\bm z}(\Gamma)|
 \leq C_\alpha.
 \label{S:eq:connected-twist-derivative-bound}
\end{equation}
For a translation class of diameter $d$, the number of embeddings into an open interval
of length $n$ is $n+O(d)$; classes with $d>n$ have an exponentially small total weight.
Separating the extensive embedding count from the boundary discrepancy gives
Eq.~\eqref{S:eq:complex-pressure}, and Eq.~\eqref{S:eq:connected-twist-derivative-bound}
gives Eq.~\eqref{S:eq:complex-pressure-derivatives}.  A fixed number of end terms changes
only the uniformly bounded boundary function.

Global $SU(2)$ invariance implies
$L_{n,\beta}(R\bm z)=L_{n,\beta}(\bm z)$ for every $R\in SO(3)$, and the same holds
for the bulk and boundary pieces.  Their gradients therefore vanish at the origin and
the bulk Hessian is a scalar matrix.  Differentiating the finite-volume logarithm twice
at zero identifies that scalar with the connected susceptibility in
Eq.~\eqref{S:eq:complex-pressure-hessian}.  The same convergent expansion makes
$\kappa_\beta$ analytic in $\beta$.  At $\beta=0$,
$\kappa_0=s(s+1)/3$, so continuity yields the stated strictly positive lower bound after
reducing the temperature interval. \hfill$\square$

\paragraph{Consequence for the two group saddles.}
Taylor's theorem and Lemma B.1 imply, for real $\bm X$ in a fixed neighborhood of the
identity,
\begin{align}
 -C\|\bm X\|^2&\leq\Re\phi_\beta(\bm X)\leq-c\|\bm X\|^2,
 &\|\nabla\phi_\beta(\bm X)\|&\leq C\|\bm X\|,\\
 |b_{n,\beta}(\bm X)|&\leq C\|\bm X\|^2,
 &\|\nabla b_{n,\beta}(\bm X)\|&\leq C\|\bm X\|.
 \label{S:eq:complex-pressure-local-bounds}
\end{align}
The boundary estimates use $b_{n,\beta}(0)=0$ and rotational invariance.  Exponentiation
therefore gives the local Gaussian upper and lower bounds and the gradient estimate used
in Lemma C.1.  Multiplication by the central element $-I$ changes only the scalar phase
$(-1)^{2sn}$, so the same estimates hold at both elements of the center.

\section{Uniform twisted-partition bounds on the group}\label{app:twisted-bounds}

\textbf{Lemma C.1 (uniform twisted-partition bounds).}
Let $\mathcal C=\{\pm I\}$ be the center of $SU(2)$ and
$r(g)=\operatorname{dist}(g,\mathcal C)$ for any fixed bi-invariant metric.  There is a $\beta_{\rm tw}>0$ such that, for $0\leq\beta<\beta_{\rm tw}$, constants $c_\beta,C_\beta>0$, independent of $n$, satisfy
\begin{align}
 |f_{n,\beta}(g)|&\leq Ce^{-cnr(g)^2},
 \label{S:eq:twistgaussian}\\
 \|\nabla f_{n,\beta}(g)\|&\leq
 C[1+nr(g)]e^{-cnr(g)^2}+Cne^{-cn}.
 \label{S:eq:twistgradient}
\end{align}
Consequently,
\begin{align}
 c n^{-3/2}\leq\int dg\,|f_{n,\beta}(g)|^2&\leq Cn^{-3/2},
 \label{S:eq:L2twist}\\
 \int dg\,|f_{n,\beta}(g)|&\leq Cn^{-3/2},
 \qquad
 \int dg\,\|\nabla f_{n,\beta}(g)\|^2\leq Cn^{-1/2}.
 \label{S:eq:L1H1twist}
\end{align}

\emph{Proof.}
Lemma B.1, applied to the half-chain Hamiltonian $H_A$, gives the local Gaussian upper
and lower bounds and the gradient estimate in a fixed neighborhood of the identity.
Because the on-site spin is irreducible,
$f_{n,\beta}(-g)=(-1)^{2sn}f_{n,\beta}(g)$, so the same estimates hold around the second
central element.  It remains to control the compact region outside those two
neighborhoods.

We next establish the complementary region directly at the partition-function level.  Retaining the interaction decomposition above, set
\begin{equation}
 r_0=\max_X|X|,
 \qquad
 J=\sup_n\frac1n\sum_{X\subset A}\|\Phi_X\|<\infty.
 \label{S:eq:interactionconstants}
\end{equation}
Finite range and bounded strength make both constants independent of $n$.  Fix the saddle radius $\delta>0$ small enough that $K_\delta=\{g:r(g)\geq\delta\}$ is nonempty.  Irreducibility and the equality condition in $|\Tr U|\leq\dim U$ give
\begin{equation}
 0<\bar q_\delta=\sup_{g\in K_\delta}
 \frac{|\chi_s(g)|}{d_s}<1,
 \qquad d_s=2s+1.
 \label{S:eq:charactercontraction}
\end{equation}
Expand the exponential in ordered interaction tuples.  If
$S=\bigcup_{a=1}^kX_a$, then the operator
$A_S=\Phi_{X_1}\cdots\Phi_{X_k}$ is supported on $S$.  Tensor-product factorization and the trace-norm bound give
\begin{equation}
 \Tr[U_n(g)A_S]=\chi_s(g)^{n-|S|}
 \Tr_S[U_S(g)A_S],
 \qquad
 |\Tr_S[U_S(g)A_S]|
 \leq d_s^{|S|}\prod_{a=1}^k\|\Phi_{X_a}\|.
 \label{S:eq:supportfactorization}
\end{equation}
Thus every site not covered by the interaction tuple supplies one strict character-contraction factor.  Therefore, for $g\in K_\delta$,
\begin{align}
 &d_s^{-n}\left|
 \Tr\!\left[U_n(g)\Phi_{X_1}\cdots\Phi_{X_k}\right]
 \right|\nonumber\\
 &\hspace{1cm}\leq
 \bar q_\delta^{n-|S|}\prod_{a=1}^k\|\Phi_{X_a}\|
 \leq
 \bar q_\delta^{n-r_0k}\prod_{a=1}^k\|\Phi_{X_a}\|.
 \label{S:eq:tuplebound}
\end{align}
Absolute summation of the exponential series and
$\sum_X\|\Phi_X\|\leq Jn$ yield
\begin{equation}
 \frac{|F_{n,\beta}(g)|}{d_s^n}
 \leq \bar q_\delta^n
 \exp\!\left(\beta Jn\bar q_\delta^{-r_0}\right).
 \label{S:eq:awaynumerator}
\end{equation}
The spectral bound $\|H_A\|\leq Jn$ gives
$Z_A\geq d_s^ne^{-\beta Jn}$.  Hence
\begin{equation}
 |f_{n,\beta}(g)|
 \leq \exp\!\left\{-n\left[
 -\log\bar q_\delta-\beta J(1+\bar q_\delta^{-r_0})
 \right]\right\}.
 \label{S:eq:awaybound}
\end{equation}
The bracket is positive after reducing $\beta_{\rm tw}$.
More explicitly, this complementary estimate holds with a positive rate whenever
\begin{equation}
 \beta<\beta_{\rm away}(\delta)
 :=\frac{-\log\bar q_\delta}
 {J(1+\bar q_\delta^{-r_0})}.
 \label{S:eq:awaythreshold}
\end{equation}

Put $C_s=\sup_{\|Y\|=1}\|dU_s(Y)\|<\infty$.  For a unit Lie-algebra vector $Y$, the product rule for $D_YU_n(g)$ marks one site, and
$|\Tr[U_s(g)dU_s(Y)]|\leq d_sC_s$.  In a term supported on $S$, a marked site in $S$ contributes at most $C_s|S|\bar q_\delta^{n-|S|}$, while a marked site in $S^c$ contributes at most
$C_s(n-|S|)\bar q_\delta^{n-|S|-1}$.  Thus, with
$t=\beta Jn\bar q_\delta^{-r_0}$ and $|S|\leq r_0k$, absolute summation gives
\begin{align}
 \frac{\|\nabla F_{n,\beta}(g)\|}{d_s^n}
 &\leq C_s\bar q_\delta^n
 \sum_{k\geq0}\frac{(\beta Jn)^k}{k!}
 \left(r_0k\bar q_\delta^{-r_0k}
 +n\bar q_\delta^{-1-r_0k}\right)\nonumber\\
 &\leq C_s\bar q_\delta^n e^t
 \left(r_0t+n\bar q_\delta^{-1}\right).
 \label{S:eq:awaygradientnumerator}
\end{align}
Dividing by $Z_A$ (which is independent of $g$) and using the same positive exponent as in Eq.~\eqref{S:eq:awaybound} gives
\begin{equation}
 \|\nabla f_{n,\beta}(g)\|
 \leq C_{\beta,\delta}n e^{-c_{\beta,\delta}n},
 \qquad g\in K_\delta.
 \label{S:eq:awaygradient}
\end{equation}
Together with Eq.~\eqref{S:eq:complex-pressure-local-bounds}, these complementary estimates imply Eqs.~\eqref{S:eq:twistgaussian} and \eqref{S:eq:twistgradient} after adjusting constants on the compact group.
The local expansion also gives a lower bound of fixed modulus in balls of radius $c_0n^{-1/2}$ about $\mathcal C$.  Haar measure is three dimensional there, $dg\asymp d^3X$; integrating the upper and lower saddle bounds proves Eqs.~\eqref{S:eq:L2twist} and \eqref{S:eq:L1H1twist}.  \hfill$\square$

\section{Character estimates and logarithmic representation yield}\label{app:characters}

Character orthogonality and Eq.~\eqref{S:eq:characterexpansion} convert these integrated estimates directly into
\begin{equation}
 \frac{D_n}{Z_A^2}=\int dg\,|f_{n,\beta}(g)|^2\asymp n^{-3/2}.
 \label{S:eq:Dbounds}
\end{equation}
Since $|\chi_j(g)|\leq d_j=2j+1$, Eq.~\eqref{S:eq:L1H1twist} also gives
\begin{equation}
 \frac{a_{n,j}}{Z_A}\leq Cd_jn^{-3/2}.
 \label{S:eq:coefficientbound}
\end{equation}
It follows that, for $0<x\leq1$,
\begin{equation}
 \Pr_{p^{(0)}}(d_j\leq x\sqrt n)
 \leq Cx^3.
 \label{S:eq:lowtail}
\end{equation}
Indeed, the numerator is bounded by
$CZ_A^2n^{-3}\sum_{d_j\leq x\sqrt n}d_j^2
\leq CZ_A^2x^3n^{-3/2}$, while Eq.~\eqref{S:eq:Dbounds} lower bounds the denominator.

For the upper side, take the bi-invariant Laplacian normalized by
$-\Delta\chi_j=j(j+1)\chi_j$.  Parseval and integration by parts yield
\begin{equation}
 \sum_j j(j+1)\left(\frac{a_{n,j}}{Z_A}\right)^2
 =\int dg\,\|\nabla f_{n,\beta}(g)\|^2\leq Cn^{-1/2}.
\end{equation}
Together with Eq.~\eqref{S:eq:Dbounds}, this proves
\begin{equation}
 \mathbb E_{p^{(0)}}[j(j+1)]\leq Cn.
 \label{S:eq:casimirmoment}
\end{equation}
Writing $X_n=d_j/\sqrt n$, Eq.~\eqref{S:eq:lowtail} gives
$\mathbb E[-\log X_n;X_n\leq1]=O(1)$ by tail integration.  On the positive side, Eq.~\eqref{S:eq:casimirmoment} and Jensen's inequality give
\begin{equation}
 \mathbb E[(\log X_n)_+]
 \leq\frac12\mathbb E\log(1+X_n^2)
 \leq\frac12\log(1+\mathbb E X_n^2)=O(1).
\end{equation}
Consequently,
\begin{equation}
 \sum_jp_j^{(0)}\log_2(2j+1)
 =\frac12\log_2 n+O_\beta(1).
 \label{S:eq:decoupledresult}
\end{equation}

\section{Stability under a cut-local Gibbs perturbation}\label{app:stability}

The following stability lemma controls the effect of restoring the interactions across the bipartition.

\textbf{Lemma E.1 (cut-local stability of the representation distribution).}
Let $H^{(0)}=H_A+H_B$ and $H=H^{(0)}+V_\partial$, where $V_\partial$ is supported in
a fixed-size neighborhood of the cut and is $SU(2)$ invariant.  Let $\sigma_0$ and
$\sigma$ be the normalized global-singlet Gibbs states of $H^{(0)}$ and $H$,
respectively.  At every fixed inverse temperature for which the QBP locality estimate is
exponential, there are constants $C,\mu,c>0$ such that, for every $R$, one can choose an
$SU(2)$-invariant operator $\eta_R$ supported within distance $R$ of the cut with
\begin{equation}
 \sigma=\frac{\eta\sigma_0\eta^\dagger}{d},
 \qquad d=\Tr(\eta\sigma_0\eta^\dagger)\geq c,
 \qquad
 \|\eta-\eta_R\|\leq Ce^{-\mu R},
 \qquad \|\eta_R\|\leq C.
 \label{S:eq:cut-stability-qbp}
\end{equation}
If $Q_{\leq L}$ and $Q_{\geq L}$ are the spectral projectors of the spin of half $A$,
then
\begin{align}
 \Tr(Q_{\leq L}\sigma)
 &\leq C\Tr(Q_{\leq L+cR}\sigma_0)+Ce^{-\mu R},
 \label{S:eq:cut-stability-low}\\
 \Tr(Q_{\geq L}\sigma)
 &\leq C\Tr(Q_{\geq L-cR}\sigma_0)+Ce^{-\mu R}.
 \label{S:eq:cut-stability-high}
\end{align}
Writing $D=2J_A+1$, one also has
\begin{equation}
 \Tr(D^2\sigma)
 \leq C\left[\Tr(D^2\sigma_0)+R^2+n^2e^{-\mu R}\right].
 \label{S:eq:cut-stability-moment}
\end{equation}

\emph{Proof.}
For $H(s)=H^{(0)}+sV_\partial$, QBP gives
\begin{equation}
 e^{-\beta H(s)}=\eta(s)e^{-\beta H^{(0)}}\eta(s)^\dagger,
 \qquad
 \frac{d\eta(s)}{ds}=-G(s)\eta(s),
 \qquad \|G(s)\|\leq\frac{\beta}{2}\|V_\partial\|.
 \label{S:eq:qbp-ode}
\end{equation}
The ordered exponential is invertible and the differential equation for its inverse gives
\begin{equation}
 \|\eta(s)\|,\ \|\eta(s)^{-1}\|
 \leq e^{\beta s\|V_\partial\|/2}.
 \label{S:eq:qbp-inverse}
\end{equation}
Every QBP generator is $SU(2)$ invariant because $H(s)$ and $V_\partial$ are invariant;
hence $[\eta,P_0]=0$.  Projecting Eq.~\eqref{S:eq:qbp-ode} onto the singlet sector and
normalizing gives Eq.~\eqref{S:eq:cut-stability-qbp}, with
$d\geq\|\eta^{-1}\|^{-2}$.

The finite-range QBP locality theorem supplies a truncation supported in the radius-$R$
neighborhood with exponentially small operator-norm error.  Haar averaging that
truncation preserves its support and norm bound, does not increase the error, and makes
it exactly $SU(2)$ invariant.

Decompose $A=A_0A_R$, where $A_R$ contains the $R$ sites nearest the cut.  If the
untouched region $A_0$ has spin $\ell$, both the incoming and outgoing total spins of
$A$ lie within the tensor-product interval generated by $\ell$ and the spins in $A_R$.
Consequently any operator supported on $A_RB_R$ has the exact bandwidth property
\begin{equation}
 P_k\eta_RP_j=0\quad\text{whenever}\quad |k-j|>cR,
 \label{S:eq:qbp-bandwidth}
\end{equation}
where $P_j$ projects onto half-chain spin $j$.  Equivalently,
\begin{equation}
 Q_{\leq L}\eta_R=Q_{\leq L}\eta_RQ_{\leq L+cR},
 \qquad
 Q_{\geq L}\eta_R=Q_{\geq L}\eta_RQ_{\geq L-cR}.
 \label{S:eq:qbp-selection-projectors}
\end{equation}
The decoupled state $\sigma_0$ commutes with every $P_j$.  Inserting the corresponding
projector on both sides of $\sigma_0$, using $\|\eta_R\|\leq C$, and comparing $\eta_R$
with $\eta$ gives Eqs.~\eqref{S:eq:cut-stability-low} and
\eqref{S:eq:cut-stability-high}.  The comparison uses, for every projector $Q$,
\begin{equation}
 \left|\Tr Q(\eta\sigma_0\eta^\dagger-
 \eta_R\sigma_0\eta_R^\dagger)\right|
 \leq\|\eta-\eta_R\|(\|\eta\|+\|\eta_R\|).
 \label{S:eq:qbp-projector-error}
\end{equation}

For the moment bound, first use Eq.~\eqref{S:eq:qbp-bandwidth}.  Whenever
$P_k\eta_RP_j\neq0$, the eigenvalues $d_j=2j+1$ and $d_k=2k+1$ obey
$d_k\leq d_j+2cR$, and hence $d_k^2\leq2d_j^2+8c^2R^2$.  Since
$[\sigma_0,P_j]=0$,
\begin{align}
 \Tr(D^2\eta_R\sigma_0\eta_R^\dagger)
 &\leq2\|\eta_R\|^2\Tr(D^2\sigma_0)
 +8c^2R^2\|\eta_R\|^2.
 \label{S:eq:qbp-banded-moment}
\end{align}
Finally $\|D^2\|=O(n^2)$, so replacing $\eta_R$ by $\eta$ costs at most
$Cn^2e^{-\mu R}$ by the same trace-norm estimate as
Eq.~\eqref{S:eq:qbp-projector-error}.  Division by the uniformly positive normalization
$d$ proves Eq.~\eqref{S:eq:cut-stability-moment}. \hfill$\square$

\paragraph{Completion of the logarithmic scaling.}
Choose $R=q\log n$ with $q$ large enough that $e^{-\mu R}\leq n^{-4}$.  If the
decoupled distribution satisfies
\begin{equation}
 \Pr_{\sigma_0}(2j+1\leq x\sqrt n)\leq Cx^3,
 \qquad
 \mathbb E_{\sigma_0}(2j+1)^2\leq Cn,
 \label{S:eq:decoupled-input-tail}
\end{equation}
then Lemma E.1 gives
\begin{equation}
 \Pr_\sigma(2j+1\leq x\sqrt n)
 \leq C\left(x+\frac{C\log n}{\sqrt n}\right)^3+n^{-4},
 \qquad
 \mathbb E_\sigma(2j+1)^2\leq Cn.
 \label{S:eq:interacting-output-tail}
\end{equation}
Tail integration controls the negative logarithmic part, while Jensen's inequality applied
to $\log(1+D^2/n)$ controls the positive part.  Thus
\begin{equation}
 \mathbb E_\sigma\log_2(2j+1)=\frac12\log_2n+O_\beta(1)
 =\frac12\log_2N+O_\beta(1).
 \label{S:eq:cut-stability-final}
\end{equation}

\section{Exact dimer entanglement and asymptotics}\label{app:dimer}

Each dimer Hilbert space is $V_0\oplus V_1$, with Boltzmann weight one on $V_0$ and $y=e^{-\beta\Delta}$ on $V_1$.  For $L$ dimers in one half, define the representation coefficients
\begin{equation}
 a_{L,j}(y)=\int_{SU(2)}dg\,\chi_j(g)^*[1+y\chi_1(g)]^L.
\end{equation}
Equivalently they obey the stable recursion
\begin{align}
 a_{\ell+1,0}&=a_{\ell,0}+y a_{\ell,1},\\
 a_{\ell+1,j}&=a_{\ell,j}+y(a_{\ell,j-1}+a_{\ell,j}+a_{\ell,j+1}),\quad j\geq1,
\end{align}
with $a_{0,j}=\delta_{j0}$.  The global-singlet block weights are
\begin{equation}
 p_j=\frac{a_{L,j}^2}{\sum_ka_{L,k}^2}.
\end{equation}
Let $\tau_{L,j}^{A,B}$ denote the normalized thermal states on the two multiplicity spaces.  Since no dimer crosses the cut, the projected thermal state has the locally flagged form
\begin{equation}
 \rho_{\beta,0}=\bigoplus_j p_j
 |\Phi_j\rangle\!\langle\Phi_j|\otimes
 \tau_{L,j}^A\otimes\tau_{L,j}^B.
 \label{S:eq:dimerflagged}
\end{equation}
Measuring the local flag $j$ and discarding the multiplicity systems gives the lower bound in Lemma A.1.  Conversely, diagonalizing each product state $\tau_{L,j}^A\otimes\tau_{L,j}^B$ provides a pure-state decomposition with average entanglement $\sum_jp_j\log_2(2j+1)$, so convexity gives the same upper bound on $\EF$.  The general inequality $\ED\leq\EF$ therefore proves
\begin{equation}
 \ED=\EF=\sum_jp_j\log_2(2j+1),
\end{equation}
in agreement with the locally distinguishable zero-angular-momentum mixtures analyzed in Refs.~\cite{Livine2005,Moharramipour2024}.

We now justify the constant term, including the tails needed to pass from a local saddle to a logarithmic expectation.  Normalize the single-dimer character as
\begin{equation}
 q_y(g)=\frac{1+y\chi_1(g)}{1+3y},
 \qquad
 b_{L,j}=\frac{a_{L,j}}{(1+3y)^L}.
 \label{S:eq:dimerq}
\end{equation}
For fixed $y>0$, $|q_y(g)|=1$ only at the two central elements.  In exponential coordinates about either center,
\begin{equation}
 q_y(\exp X)=1-\kappa\|X\|^2+O(\|X\|^4),
 \qquad
 \kappa(\beta)=\frac{y}{1+3y}.
\end{equation}
To evaluate the character coefficient, use the radial Haar measure
$dg=c\sin^2(\theta/2)d\theta$ and
$\chi_j(\theta)=\sin[(2j+1)\theta/2]/\sin(\theta/2)$.  Split the integral into fixed neighborhoods of the two centers and their complement.  The complement is exponentially small in $L$; in the two neighborhoods, rescaling $\theta=t/\sqrt L$ and applying dominated convergence gives the three-dimensional Fourier saddle.  Thus, for every fixed $A<\infty$, uniformly for $0\leq j\leq A\sqrt L$,
\begin{equation}
 b_{L,j}=C_yL^{-3/2}(2j+1)
 \exp\!\left[-\frac{j(j+1)}{4\kappa L}\right][1+o(1)].
 \label{S:eq:saddle}
\end{equation}
Here $C_y>0$ is independent of $j$ and $L$; its value will cancel from normalized probabilities.

The same local Gaussian and away-from-center estimates provide the global tail controls needed for the logarithmic observable.  Character orthogonality gives
\begin{align}
 B_L:=\sum_jb_{L,j}^2
 &=\int dg\,|q_y(g)|^{2L}\asymp L^{-3/2},
 \label{S:eq:dimerBL}\\
 0\leq b_{L,j}&\leq C_y'(2j+1)L^{-3/2}.
 \label{S:eq:dimercoefficient}
\end{align}
Moreover, applying the group Laplacian and using
$\nabla q_y(\exp X)=O(\|X\|)$ at the two saddles yields
\begin{equation}
 \sum_j j(j+1)b_{L,j}^2
 =\int dg\,\|\nabla q_y(g)^L\|^2
 \leq C_yL^{-1/2}.
 \label{S:eq:dimermomentraw}
\end{equation}
Equations~\eqref{S:eq:dimerBL}--\eqref{S:eq:dimermomentraw} imply, for the probability distribution $p_j=b_{L,j}^2/B_L$,
\begin{align}
 \Pr(2j+1\leq x\sqrt L)&\leq C_yx^3,
 \qquad 0<x\leq1,
 \label{S:eq:dimerlowtail}\\
 \mathbb E[j(j+1)]&\leq C_yL.
 \label{S:eq:dimermoment}
\end{align}
The first inequality follows by summing $(2j+1)^2$ up to $x\sqrt L$, exactly as in Eq.~\eqref{S:eq:lowtail}.

Define the strictly positive rescaled variable
\begin{equation}
 U_L=\frac{2j+1}{\sqrt{8\kappa L}}.
\end{equation}
For any fixed compact interval in $u$, Eq.~\eqref{S:eq:saddle} and a Riemann sum give a density proportional to $u^2e^{-u^2}$.  Equation~\eqref{S:eq:dimermoment} makes the probability outside $U_L\leq A$ uniformly $O(A^{-2})$; sending first $L\to\infty$ and then $A\to\infty$ fixes the normalized limiting density to
\begin{equation}
 h(u)=\frac{4}{\sqrt\pi}u^2e^{-u^2},
 \qquad u>0.
 \label{S:eq:dimerdensity}
\end{equation}
It remains to justify using the unbounded function $\log u$.  Equation~\eqref{S:eq:dimerlowtail} and tail integration give, for sufficiently small $\epsilon>0$,
\begin{equation}
 \sup_L\mathbb E[-\log U_L;U_L<\epsilon]
 \leq C_y\epsilon^3(1+|\log\epsilon|)
 \xrightarrow[\epsilon\downarrow0]{}0.
 \label{S:eq:dimernegativelog}
\end{equation}
For the positive tail, Eq.~\eqref{S:eq:dimermoment} implies $\sup_L\mathbb E U_L^2<\infty$, and Markov's inequality gives, for $M>1$,
\begin{equation}
 \sup_L\mathbb E[\log U_L;U_L>M]
 \leq C_y\frac{1+\log M}{M^2}\xrightarrow[M\to\infty]{}0.
\end{equation}
Thus $\log U_L$ is uniformly integrable, and weak convergence in Eq.~\eqref{S:eq:dimerdensity} implies
\begin{equation}
 \lim_{L\to\infty}\mathbb E\log U_L
 =\frac12\psi(3/2)=1-\frac\gamma2-\ln2.
\end{equation}
Finally,
$\log_2(2j+1)=\frac12\log_2(8\kappa L)+\log_2U_L$, which gives
\begin{equation}
 \ED=\frac12\log_2L+\frac12\log_2(2\kappa)
 +\frac{1-\gamma/2}{\ln2}+o(1).
 \label{S:eq:constant}
\end{equation}
In the dilute low-temperature expansion the half-chain triplet activity before the global singlet projection is $Le^{-\beta\Delta}$; the lowest projected entangled contribution contains one triplet in each half and consequently has weight of order $(Le^{-\beta\Delta})^2$.  The crossover is therefore fixed by $Le^{-\beta\Delta}\sim1$ and
$T_*(N)\sim\Delta/\ln N$.

\section{Exact dilute-activity crossover}\label{app:crossover}

The finite-size collapse converges to a closed analytic scaling function that describes the full activation crossover.

\textbf{Theorem G.1 (exact crossover function).}
Let
\begin{equation}
 Y_L(y)=\frac{\sum_{j=0}^{L}a_{L,j}(y)^2\log_2(2j+1)}
 {\sum_{j=0}^{L}a_{L,j}(y)^2},
 \qquad
 a_{L,j}(y)=\int_{SU(2)}dg\,\chi_j(g)[1+y\chi_1(g)]^L.
 \label{S:eq:crossover-finite}
\end{equation}
For every $x\geq0$,
\begin{equation}
 \lim_{L\to\infty}Y_L(x/L)=\mathcal F(x),
 \label{S:eq:crossover-limit}
\end{equation}
locally uniformly in $x$, where
\begin{align}
 A_j(x)&=\int_{SU(2)}dg\,\chi_j(g)e^{x\chi_1(g)}
 =e^x\left[I_j(2x)-I_{j+1}(2x)\right],
 \label{S:eq:crossover-Aj}\\
 p_j(x)&=\frac{A_j(x)^2}{\sum_{k\geq0}A_k(x)^2}
 =\frac{[I_j(2x)-I_{j+1}(2x)]^2}
 {I_0(4x)-I_1(4x)},
 \label{S:eq:crossover-pj}\\
 \mathcal F(x)&=\sum_{j=0}^{\infty}p_j(x)\log_2(2j+1).
 \label{S:eq:crossover-function}
\end{align}
Here $I_j$ denotes the modified Bessel function.  The two asymptotic regimes are
\begin{align}
 \mathcal F(x)&=x^2\log_2 3+O(x^3), &&x\downarrow0,
 \label{S:eq:crossover-smallx}\\
 \mathcal F(x)&=\frac12\log_2(2x)
 +\frac{1-\gamma/2}{\ln2}+o(1), &&x\to\infty.
 \label{S:eq:crossover-largex}
\end{align}

\emph{Proof.}
Let $\nu_{m,j}$ be the multiplicity of $V_j$ in $V_1^{\otimes m}$, so that
$\chi_1(g)^m=\sum_j\nu_{m,j}\chi_j(g)$.  Expanding the finite-$L$ character gives
\begin{equation}
 a_{L,j}(x/L)=\sum_{m=0}^{L}\binom{L}{m}(x/L)^m\nu_{m,j}.
 \label{S:eq:crossover-binomial}
\end{equation}
For every fixed $m$, the binomial prefactor tends to $x^m/m!$, and therefore
$a_{L,j}(x/L)\to A_j(x)$.  This convergence also passes through the normalized
logarithmic expectation.  Indeed, $\nu_{m,j}=0$ for $j>m$ and
$0\leq\nu_{m,j}\leq3^m$, which gives, uniformly for $x$ in a compact interval
$[0,X]$,
\begin{equation}
 0\leq a_{L,j}(x/L),A_j(x)
 \leq\sum_{m\geq j}\frac{(3X)^m}{m!}.
 \label{S:eq:crossover-factorial-tail}
\end{equation}
The right-hand side has a factorial tail, and its square remains summable after
multiplication by $1+\log(2j+1)$.  Dominated convergence therefore proves
Eq.~\eqref{S:eq:crossover-limit}, locally uniformly in $x$.

For the closed form, parametrize conjugacy classes by $\theta\in[0,2\pi]$ with
$dg=\pi^{-1}\sin^2(\theta/2)d\theta$,
$\chi_j(\theta)=\sin[(2j+1)\theta/2]/\sin(\theta/2)$, and
$\chi_1(\theta)=1+2\cos\theta$.  Then
\begin{align}
 A_j(x)
 &=\frac{e^x}{2\pi}\int_0^{2\pi}d\theta\,
 e^{2x\cos\theta}\left[\cos(j\theta)-\cos((j+1)\theta)\right],
\end{align}
which is Eq.~\eqref{S:eq:crossover-Aj}.  Parseval gives
$\sum_jA_j(x)^2=\int dg\,e^{2x\chi_1(g)}
=e^{2x}[I_0(4x)-I_1(4x)]$, proving Eq.~\eqref{S:eq:crossover-pj}.

At small $x$, $A_0(x)=1+O(x^2)$, $A_1(x)=x+O(x^2)$, and
$A_j(x)=O(x^2)$ for $j\geq2$, which gives Eq.~\eqref{S:eq:crossover-smallx}.
For large $x$, the two-center group saddle, equivalently the uniform Bessel asymptotics,
gives for $j=O(\sqrt x)$
\begin{equation}
 e^{-3x}A_j(x)=C x^{-3/2}(2j+1)
 \exp\left[-\frac{j(j+1)}{4x}\right][1+o(1)],
 \label{S:eq:crossover-saddle}
\end{equation}
uniformly on compact intervals of $j/\sqrt x$.  The same character and group-Laplacian
bounds used for the fixed-temperature saddle control the low-spin and high-spin tails.
Thus
\begin{equation}
 U_x=\frac{2j+1}{\sqrt{8x}}
 \quad\Longrightarrow\quad
 h(u)=\frac4{\sqrt\pi}u^2e^{-u^2},\qquad u>0,
 \label{S:eq:crossover-maxwell}
\end{equation}
and $\log U_x$ is uniformly integrable.  Since
$\mathbb E_h\log U=\frac12\psi(3/2)=1-\gamma/2-\ln2$,
Eq.~\eqref{S:eq:crossover-largex} follows. \hfill$\square$

Because $x=Le^{-\Delta/T}$, the transition from the perturbative regime
$\mathcal F(x)\sim x^2\log_2 3$ to an order-one entanglement resource occurs at
$x=O(1)$.  Equivalently,
\begin{equation}
 \frac{\Delta}{T_*(L)}=\ln L+O(1),
 \qquad
 T_*(N)=\frac{\Delta}{\ln N+O(1)}
 \quad(N=4L).
 \label{S:eq:crossover-scale-refined}
\end{equation}

\section{Exact-diagonalization basis and validation}\label{app:ed}

We used the open $J_1$--$J_2$ Hamiltonian with $J_2/J_1=0.37$.  For $N\leq14$ we constructed the $S^z=0$ computational basis, diagonalized total $\bm S^2$, and retained its kernel.  For $N=16$ we instead used the linearly independent noncrossing valence-bond basis; the Hamiltonian and half-chain $\bm S_A^2$ were solved as generalized eigenvalue problems with the valence-bond Gram matrix.  For every energy eigenstate in the singlet sector we evaluated the projectors onto half-chain $j$, thermally averaged them, and formed
\begin{equation}
 Y_N=\sum_jp_j\log_2(2j+1).
\end{equation}
The two half-chain parities have visibly different finite-size corrections. Figure~\ref{fig:ed} probes the theorem-fixed coefficient through the residual $Y_N-\frac12\log_2N$ for both parities and probes the predicted representation width through the cumulative distributions of $j/\sqrt N$. The numerical tables used for Fig.~\ref{fig:ed} are included in the accompanying reproduction package.

\end{document}